\documentclass[envcountsame]{llncs}

\usepackage{amsmath,amssymb}
\usepackage{tikz}
\usepackage{graphicx,subfigure}
\usepackage[pdftitle={Graph drawings with one bend and few slopes},
            pdfauthor={Kolja Knauer, Bartosz Walczak},
            bookmarks=false,
            colorlinks=true,linkcolor=black,citecolor=black,filecolor=black,urlcolor=black]{hyperref}

\let\doendproof\endproof
\renewcommand\endproof{\qed\doendproof}

\let\leq\leqslant
\let\geq\geqslant
\let\epsilon\varepsilon

\allowdisplaybreaks[1]

\begin{document}

\title{Graph Drawings with One Bend and Few Slopes}

\author{Kolja Knauer\inst{1}\fnmsep\thanks{Supported by ANR EGOS grant ANR-12-JS02-002-01 and PEPS grant EROS.}\and Bartosz Walczak\inst{2}\fnmsep\thanks{Supported by MNiSW grant 911/MOB/2012/0.}}
\institute{Aix-Marseille Universit\'e, CNRS, LIF UMR 7279, Marseille, France\\\email{kolja.knauer@lif.univ-mrs.fr}\and Theoretical Computer Science Department, Faculty of Mathematics and Computer~Science, Jagiellonian University, Krak\'ow, Poland\\\email{walczak@tcs.uj.edu.pl}}

\maketitle
\thispagestyle{plain}

\begin{abstract}
We consider drawings of graphs in the plane in which edges are represented by polygonal paths with at most one bend and the number of different slopes used by all segments of these paths is small.
We prove that $\lceil\frac{\Delta}{2}\rceil$ edge slopes suffice for outerplanar drawings of outerplanar graphs with maximum degree $\Delta\geq 3$.
This matches the obvious lower bound.
We also show that $\lceil\frac{\Delta}{2}\rceil+1$ edge slopes suffice for drawings of general graphs, improving on the previous bound of $\Delta+1$.
Furthermore, we improve previous upper bounds on the number of slopes needed for planar drawings of planar and bipartite planar graphs.
\end{abstract}

\section{Introduction}

A \emph{one-bend drawing} of a graph $G$ is a mapping of the vertices of $G$ into distinct points of the plane and of the edges of $G$ into polygonal paths each consisting of at most two segments joined at the \emph{bend} of the path, such that the polygonal paths connect the points representing their end-vertices and pass through no other points representing vertices nor bends of other paths.
If it leads to no confusion, in notation and terminology, we make no distinction between a vertex and the corresponding point, and between an edge and the corresponding path.
The \emph{slope} of a segment is the family of all straight lines parallel to this segment.
The \emph{one-bend slope number} of a graph $G$ is the smallest number $s$ such that there is a one-bend drawing of $G$ using $s$ slopes.
Similarly, one defines the \emph{planar one-bend slope number} and the \emph{outerplanar one-bend slope number} of a planar and respectively outerplanar graphs if the drawing additionally has to be planar and respectively outerplanar.
Since at most two segments at each vertex can use the same slope, $\lceil\frac{\Delta}{2}\rceil$ is a lower bound on the one-bend slope number.
Here and further on, $\Delta$ denotes the maximum degree of the graph considered.

\subsection{Results}

Our main contribution (Theorem~\ref{thm:outerplanar}) is that the outerplanar one-bend slope number of every outerplanar graph is equal to the above-mentioned obvious lower bound of $\lceil\frac{\Delta}{2}\rceil$ except for graphs with $\Delta=2$ that contain cycles, which need $2$ slopes.
For general graphs, we show that every graph admits a one-bend drawing using at most $\lceil\frac{\Delta}{2}\rceil+1$ slopes (Theorem~\ref{thm:general}), which improves on the upper bound of $\Delta+1$ shown in~\cite{Duj-07a}.

For planar graphs, it was shown in~\cite{Kes-13} that the planar one-bend slope number is always at most $2\Delta$.
In the same paper, it was shown that sometimes $\frac{3}{4}(\Delta-1)$ slopes are necessary.
We improve the upper bound to $\frac{3}{2}\Delta$ (Proposition~\ref{prop:planar}) and bound the planar one-bend slope number of planar bipartite graphs by $\Delta+1$ (Proposition~\ref{prop:bipplanar}).
We also show that there are planar bipartite graphs requiring $\frac{2}{3}(\Delta-1)$ slopes in any planar one-bend drawing (Proposition~\ref{prop:lowbipplanar}).
Every planar graph admits a planar $2$-bend drawing with $\lceil\frac{\Delta}{2}\rceil$ slopes~\cite{Kes-13}.

Apart from improving upon earlier results, one of our motivations for studying the one-bend slope number is that it arises as a relaxation of the slope number, a parameter extensively studied in recent years.
The one-bend slope number also naturally generalizes problems concerning one-bend orthogonal drawings, which have been of interest in the graph drawing community over the past years.
We continue with a short overview of these studies.

\subsection{Related Results: Slope Number}

The \emph{slope number} of a graph $G$, introduced by Wade and Chu~\cite{Wad-94}, is the smallest number $s$ such that there is a \emph{straight-line drawing} of $G$ using $s$ slopes.
As for the one-bend slope number, $\lceil\frac{\Delta}{2}\rceil$ is an obvious lower bound on the slope number.
Dujmovi\'c and Wood~\cite{Duj-04} asked whether the slope number can be bounded from above by a function of the maximum degree.
This was answered independently by Bar\'at, Matou\v{s}ek and Wood~\cite{Bar-06} and by Pach and P\'alv\"olgyi~\cite{Pac-06} in the negative: graphs with maximum degree $5$ can have arbitrarily large slope number.
Dujmovi\'c, Suderman and Wood~\cite{Duj-07a} further showed that for all $\Delta\geq 5$ and sufficiently large $n$, there exists an $n$-vertex graph with maximum degree $\Delta$ and slope number at least $n^{\frac{1}{2}-\frac{1}{\Delta-2}-o(1)}$.
On the other hand, Mukkamala and P\'alv\"olgyi~\cite{Muk-12} proved that graphs with maximum degree $3$ have slope number at most $4$, improving earlier results of Keszegh, Pach, P\'alv\"olgyi and T\'oth~\cite{Kes-08} and of Mukkamala and Szegedy~\cite{Muk-09}.
The question whether graphs with maximum degree $4$ have slope number bounded by a constant remains open.

The situation is different for \emph{planar} straight-line drawings.
It is well known that every planar graph admits a planar straight-line drawing.
The \emph{planar slope number} of a planar graph $G$ is the smallest number $s$ such that there is a planar straight-line drawing of $G$ using $s$ slopes.
This parameter was first studied by Duj\-mo\-vi\'c, Eppstein, Suderman and Wood~\cite{Duj-07b} in relation to the number of vertices.
They also asked whether the planar slope number of a planar graph is bounded in terms of its maximum degree.
Jel\'{\i}nek, Jel\'{\i}nkov\'a, Kratochv\'{\i}l, Lidick\'y, Tesa\v{r} and Vysko\v{c}il~\cite{Jel-13} gave an upper bound of $O(\Delta^5)$ for planar graphs of treewidth at most $3$.
Lenhart, Liotta, Mondal and Nishat~\cite{Len-13} showed that the maximum planar slope number of a graph of treewidth at most $2$ lies between $\Delta$ and $2\Delta$.
Di Giacomo, Liotta and Montecchiani~\cite{DiG-14} showed that subcubic planar graphs with at least $5$ vertices have planar slope number at most $4$.
The problem has been solved in full generality by Keszegh, Pach and P\'alv\"olgyi~\cite{Kes-13}, who showed (with a non-constructive proof) that the planar slope number is bounded from above by an exponential function of the maximum degree.
It is still an open problem whether this can be improved to a polynomial upper bound.

Knauer, Micek and Walczak~\cite{Kna-14} showed that every outerplanar graph with $\Delta\geq 4$ has an outerplanar straight-line drawing using at most $\Delta-1$ slopes and this bound is best possible.
For outerplanar graphs with $\Delta=2$ or $\Delta=3$, the optimal upper bound is $3$.

\subsection{Related Results: Orthogonal Drawings}

Drawings of graphs that use only the horizontal and the vertical slopes are called \emph{orthogonal}.
Every drawing with two slopes can be made orthogonal by a simple affine transformation of the plane.
Felsner, Kaufmann and Valtr~\cite{Fel-14} proved that a graph $G$ with $\Delta\leq 4$ admits a one-bend orthogonal drawing if and only if every induced subgraph $H$ of $G$ satisfies $E(H)\leq 2V(H)-2$.
Since outerplanar graphs satisfy the latter condition, it follows that every outerplanar graph with $\Delta\leq 4$ admits a one-bend orthogonal drawing (our Theorem~\ref{thm:outerplanar} gives an outerplanar one-bend orthogonal drawing).
Biedl and Kant~\cite{Bie-98} and Liu, Morgana and Simeone~\cite{Liu-98} showed that every planar graph with $\Delta\leq 4$ has a planar $2$-bend orthogonal drawing with the only exception of the octahedron, which has a planar $3$-bend orthogonal drawing.
Kant~\cite{Kan-96} showed that every planar graph with $\Delta\leq 3$ has a planar one-bend orthogonal drawing with the only exception of $K_4$.

\section{Outerplanar Graphs}

The main contribution of this section is to show the following:

\begin{theorem}\label{thm:outerplanar}
Every outerplanar graph with maximum degree\/ $\Delta$ admits an outerplanar one-bend drawing using at most\/ $\max\{\lceil\frac{\Delta}{2}\rceil,2\}$ slopes.
Furthermore, the set of slopes can be prescribed arbitrarily.
\end{theorem}

The structure of the proof of Theorem~\ref{thm:outerplanar} will follow the same recursive decomposition of an outerplanar graph into \emph{bubbles} that was used in~\cite{Kna-14} in the proof that every outerplanar graph has a straight-line outerplanar drawing using at most $\Delta-1$ slopes.
Although this decomposition is very natural, for completeness we present it in detail recalling definitions and lemmas from~\cite{Kna-14}.

Let $G$ be an outerplanar graph provided together with its arbitrary outerplanar drawing in the plane.
The drawing determines the cyclic order of edges at each vertex and identifies the \emph{outer face} (which is unbounded and contains all vertices on its boundary) and the \emph{inner faces} of $G$.
The edges on the boundary of the outer face are \emph{outer edges}, and all remaining ones are \emph{inner edges}.
A \emph{snip} is a simple closed counterclockwise-oriented curve $\gamma$ which
\begin{itemize}
\item passes through some pair of vertices $u$ and $v$ of $G$ (possibly being the same vertex) and through no other vertex of $G$,
\item on the way from $v$ to $u$ goes entirely through the outer face of $G$ and crosses no edge of $G$,
\item on the way from $u$ to $v$ (considered only when $u\ne v$) goes through inner faces of $G$ possibly crossing some inner edges of $G$ that are not incident to $u$ or $v$, each at most once,
\item crosses no edge of $G$ incident to $u$ or $v$ at a point other than $u$ or $v$.
\end{itemize}
Every snip $\gamma$ defines a \emph{bubble} $H$ in $G$ as the subgraph of $G$ induced by the vertices lying on or inside $\gamma$.
Since $\gamma$ crosses no outer edges, $H$ is a connected induced subgraph of $G$.
The \emph{roots} of $H$ are the vertices $u$ and $v$ together with all vertices of $H$ adjacent to $G-H$.
The snip $\gamma$ breaks the cyclic clockwise order of the edges of $H$ around each root of $H$ making it a linear order, which we envision as going from left to right.
We call the first edge in this order \emph{leftmost} and the last one \emph{rightmost}. 
The \emph{root-path} of $H$ is the simple oriented path $P$ in $H$ that starts at $u$ with the rightmost edge, continues counterclockwise along the boundary of the outer face of $H$, and ends at $v$ with the leftmost edge.
If $u=v$, then the root-path consists of that single vertex only.
All roots of $H$ lie on the root-path---their sequence in the order along the root-path is the \emph{root-sequence} of $H$.
A \emph{$k$-bubble} is a bubble with $k$ roots.
See Fig.~\ref{fig:bubble} for an illustration.

\begin{figure}[t]
\begin{center}
\includegraphics[width = \textwidth]{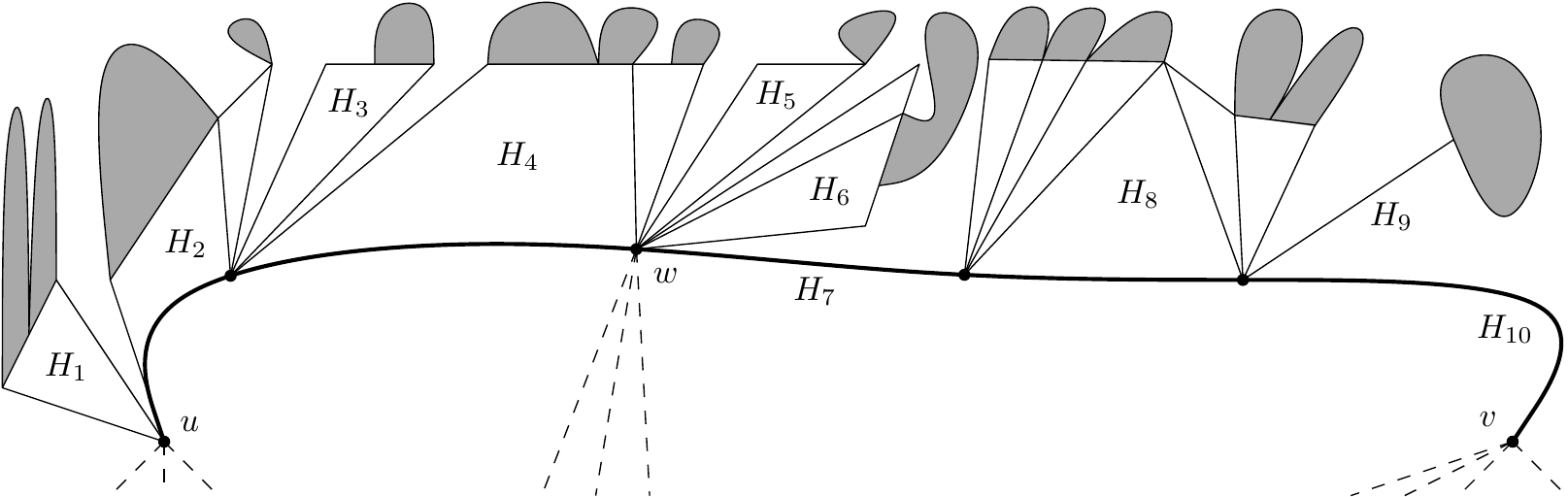}
\caption{A $3$-bubble $H$ with root-path drawn thick, root-sequence $(u,w,v)$ (connected to the remaining graph by dashed edges), and splitting sequence $(H_1,\ldots,H_{10})$, in which $H_1$, $H_3$, $H_5$, $H_6$, $H_9$ are v-bubbles and $H_2$, $H_4$, $H_7$, $H_8$, $H_{10}$ are e-bubbles.}
\label{fig:bubble}
\end{center}
\end{figure}

Except at the very end of the proof where we regard the entire $G$ as a bubble, we deal with bubbles $H$ whose first root $u$ and last root $v$ are adjacent to $G-H$.
For such bubbles $H$, all the roots, the root-path, the root-sequence and the left-to-right order of edges at every root do not depend on the particular snip $\gamma$ used to define $H$.
Specifically, for such bubbles $H$, the roots are exactly the vertices adjacent to $G-H$, while the root-path consists of the edges of $H$ incident to inner faces of $G$ that are contained in the outer face of $H$.
From now on, we will refer to the roots, the root-path, the root-sequence and the left-to-right order of edges at every root of a bubble $H$ without specifying the snip $\gamma$ explicitly.

\begin{lemma}[{\cite[Lemma 1]{Kna-14}}]\label{lem:split}
Let\/ $H$ be a bubble with root-path\/ $v_1\ldots v_k$.
Every component of\/ $H-\{v_1,\ldots,v_k\}$ is adjacent to either one vertex among\/ $v_1,\ldots,v_k$ or two consecutive vertices from\/ $v_1,\ldots,v_k$.
Moreover, there is at most one component adjacent to\/ $v_i$ and\/ $v_{i+1}$ for\/ $1\leq i<k$.
\end{lemma}

Lemma~\ref{lem:split} allows us to assign each component of $H-\{v_1,\ldots,v_k\}$ to a vertex of $P$ or an edge of $P$ so that every edge is assigned at most one component.
For a component $C$ assigned to a vertex $v_i$, the graph induced by $C\cup\{v_i\}$ is called a \emph{v-bubble}.
Such a v-bubble is a $1$-bubble with root $v_i$. 
For a component $C$ assigned to an edge $v_iv_{i+1}$, the graph induced by $C\cup\{v_i,v_{i+1}\}$ is called an \emph{e-bubble}.
Such an e-bubble is a $2$-bubble with roots $v_i$ and $v_{i+1}$.
If no component is assigned to an edge of $P$, then we let that edge alone be a \emph{trivial e-bubble}. 
All v-bubbles of $v_i$ in $H$ are naturally ordered by their clockwise arrangement around $v_i$ in the drawing.
All this leads to a decomposition of the bubble $H$ into a sequence $(H_1,\ldots,H_b)$ of v- and e-bubbles such that the naturally ordered v-bubbles of $v_1$ precede the e-bubble of $v_1v_2$, which precedes the naturally ordered v-bubbles of $v_2$, and so on.
We call it the \emph{splitting sequence} of $H$. 
The splitting sequence of a single-vertex $1$-bubble is empty.
Every $1$-bubble with more than one vertex is a v-bubble or a bouquet of several v-bubbles.
The splitting sequence of a $2$-bubble may consist of several v- and e-bubbles.
Again, see Fig.~\ref{fig:bubble} for an illustration.
 
\begin{figure}[t]
\begin{center}
\includegraphics[width=\textwidth]{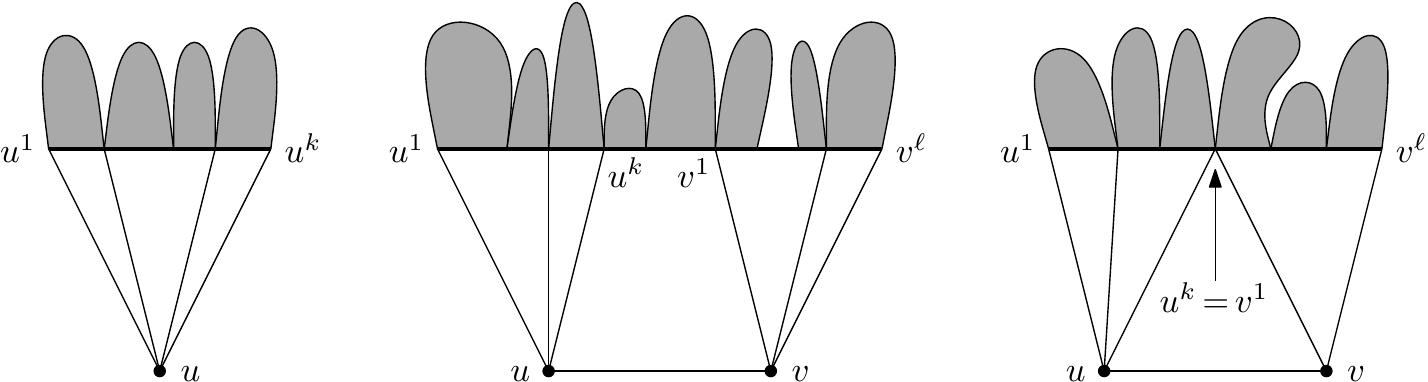}
\caption{Various ways of obtaining smaller bubbles from v- and e-bubbles described in Lemma~\ref{lem:bubble}.
The new bubbles are grayed, and the new root-paths are drawn thick.}
\label{fig:nextbubble}
\end{center}
\end{figure}

The following lemma provides the base for the recursive structure of the proof of Theorem~\ref{thm:outerplanar}.
See Fig.~\ref{fig:nextbubble} for an illustration.

\begin{lemma}[{\cite[Lemma 2, statements 2.1 and 2.3]{Kna-14}}]\label{lem:bubble}
\begin{enumerate}
\item\label{item:v-bubble}
Let\/ $H$ be a v-bubble rooted at\/ $u$. 
Let\/ $u^1,\ldots,u^k$ be the neighbors of\/ $u$ in\/ $H$ from left to right. 
Then\/ $H-\{u\}$ is a bubble with root-sequence\/ $(u^1,\ldots,u^k)$.
\item\label{item:e-bubble}
Let\/ $H$ be an e-bubble with roots\/ $u$ and\/ $v$. 
Let\/ $u^1,\ldots,u^k,v$ and\/ $u,v^1,\ldots,v^\ell$ be respectively the neighbors of\/ $u$ and\/ $v$ in\/ $H$ from left to right.
Then\/ $H-\{u,v\}$ is a bubble with root-sequence\/ $(u^1,\ldots,u^k,\allowbreak v^1,\ldots,v^\ell)$ in which\/ $u^k$ and\/ $v^1$ coincide if the inner face of\/ $H$ containing\/ $uv$ is a triangle.
\end{enumerate}
\end{lemma}

\begin{proof}[Theorem~\ref{thm:outerplanar}]
We fix $s\geq 2$, assume to be given an outerplanar graph $G$ with maximum degree $\Delta\leq 2s$, and construct an outerplanar one-bend drawing of $G$ with a prescribed set of $s$ slopes.
Actually, for most of the proof, we assume $s\geq 3$.
The case $s=2$ is sketched at the very end of the proof.

Let $D$ denote the set of $2s$ \emph{directions}, that is, oriented slopes from the prescribed set of $s$ slopes.
For a direction $d\in D$, let $d^-$ and $d^+$ denote respectively the previous and the next directions in the clockwise cyclic order on $D$.

We can assume without loss of generality that every vertex of $G$ has degree either $1$ or $2s$.
Indeed, we can raise the degree of any vertex by connecting it to new vertices of degree $1$ placed in the outer face.
With this assumption, at each vertex $u$, the direction in which one edge leaves $u$ determines the directions of the other edges at $u$.
When a vertex $u$ has all edge directions determined, we write $d(uv)$ to denote the direction determined for an edge $uv$ at $u$.

For an edge $uv$ drawn as a union of two segments $ux$ and $xv$ and for two directions $d_v,d_u\in D$ consecutive in the clockwise order on $D$, let $Q(uv,d_u,d_v)$ denote the quadrilateral $uxvy$, where $y$ is the intersection point of the rays going out of $u$ and $v$ in directions $d_u$ and $d_v$, respectively.
We express the condition that the point $y$ exists saying that the quadrilateral is \emph{well defined}.

First, consider the setting of Lemma~\ref{lem:bubble} statement~\ref{item:e-bubble}.
Assume that the edge $uv$ is the only predrawn part of $H$.
Assume further that two \emph{leading directions} $d_v,d_u\in D$ that are consecutive in the clockwise order on $D$ and have the following properties are provided:
\begin{enumerate}
\item[a.] $-d_u\notin\{d(uu^1),\ldots,d(uu^k)\}$ and $-d_v\notin\{d(vv^1),\ldots,d(vv^\ell)\}$,
\item[b.] no part of the graph other than the edge $uv$ and some short initial parts of other edges at $u$ and $v$ is predrawn in the $\epsilon$-neighborhood $Q_\epsilon$ of the quadrilateral $Q=Q(uv,d_u,d_v)$, for some sufficiently small $\epsilon>0$.
\end{enumerate}
We call $Q$ the \emph{target quadrilateral} for $H$.
We will draw $H$ in $Q_\epsilon$ in a way that will guarantee that $H$ does not cross the predrawn parts of the graph.
To this end, we need to draw the edges $uu^1,\ldots,uu^k,\allowbreak vv^1,\ldots,vv^\ell$ and the bubble $H'=H-\{u,v\}$ obtained in the conclusion of Lemma~\ref{lem:bubble} statement~\ref{item:e-bubble}.

\begin{figure}[t]
\begin{center}
\includegraphics[width=.9\textwidth]{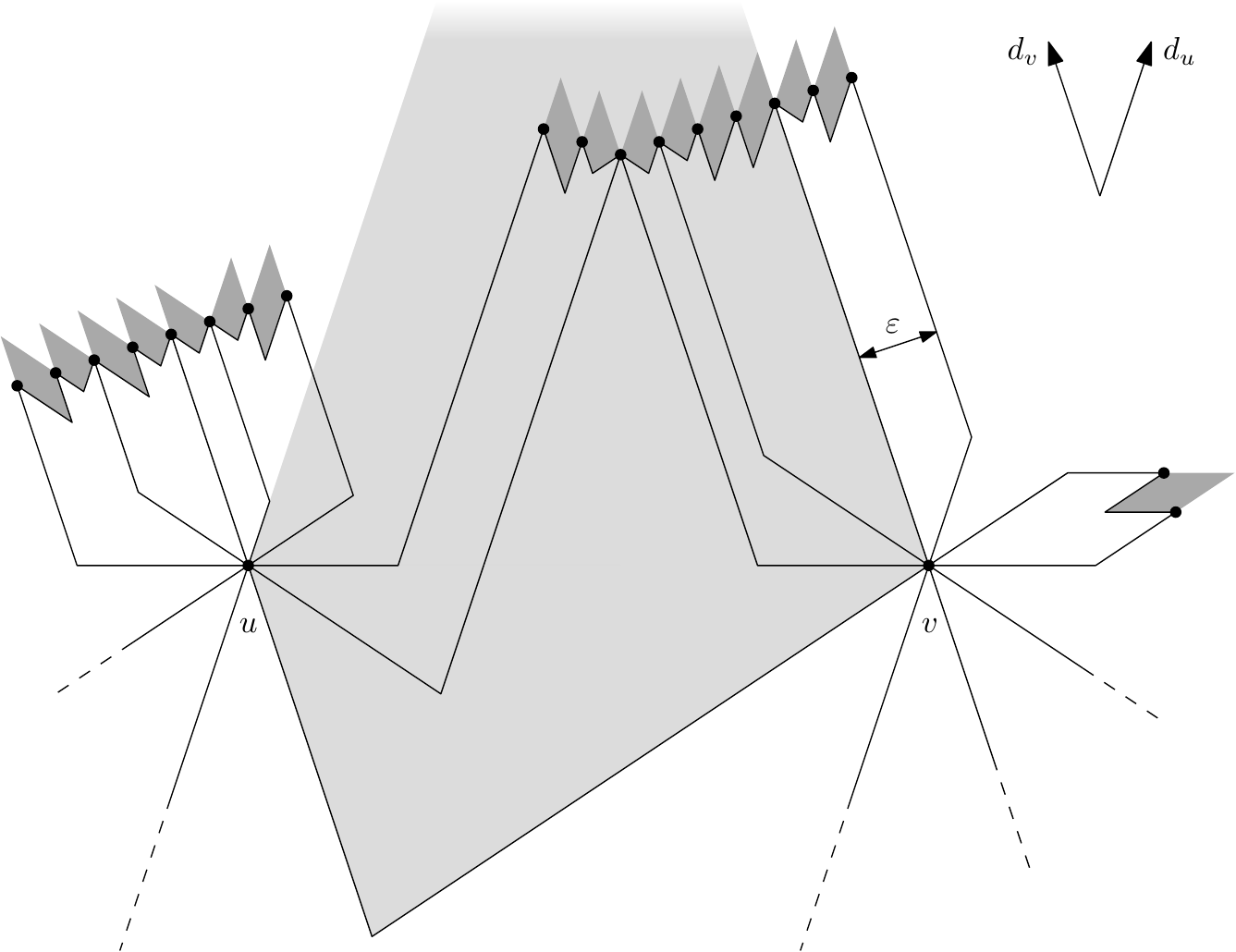}
\caption{Drawing bubbles: a v-bubble of Case~3 (left), an e-bubble (middle), and a v-bubble of Case~2 (right).
The directions $d_u$ and $d_v$ used to draw the e-bubble are also shown.
The target quadrilaterals for e-bubbles are grayed.}
\label{fig:drawing}
\end{center}
\end{figure}

The edges $uu^1,\ldots,uu^k,\allowbreak vv^1,\ldots,vv^\ell$ and the root-path $P$ of $H'$ are drawn in $Q_\epsilon$ in such a way that the following conditions are satisfied:
\begin{itemize}
\item each edge $uu^i$ leaves $u$ in direction $d(uu^i)$, bends shortly after (but far~enough to avoid crossing other edges at $u$), and continues to $u^i$ in direction $d_u$,
\item each edge $vv^i$ leaves $v$ in direction $d(vv^i)$, bends shortly after (but far~enough to avoid crossing other edges at $v$), and continues to $v^i$ in direction $d_v$,
\item each edge $xy$ of $P$ leaves $x$ in direction $-d_v^-$ if $x\in\{v^1,\ldots,v^{\ell-1}\}$ or $-d_v$ otherwise, and leaves $y$ in direction $-d_u^+$ if $y\in\{u^2,\ldots,u^k\}$ or $-d_u$ otherwise,
\item for every edge $xy$ of $P$, the quadrilateral $Q(xy,d_u,d_v)$ is well defined.
\end{itemize}
Figure~\ref{fig:drawing} illustrates how to achieve such a drawing.
As a consequence, $d_v$ and $d_u$ can be assigned as leading directions to the e-bubbles of the splitting sequence of $H'$, because (a) at their roots, the directions $-d_v$ and $-d_u$ are occupied by edges of the root-path of $H'$ or by edges going to $u$ and $v$, and (b) their target quadrilaterals are pairwise disjoint except at their common vertices and are contained in $Q_\epsilon$ far enough from $u$ and $v$.
The drawing of $H$ is completed by drawing all bubbles of the splitting sequence of $H'$ recursively.

Now, consider the setting of Lemma~\ref{lem:bubble} statement~\ref{item:v-bubble}.
Assume that the vertex $u$ is the only predrawn part of $H$.
For $\epsilon>0$ as small as necessary, we will draw $H$ in the $\epsilon$-neighborhood of the cone at $u$ spanned clockwise between the rays in directions $d(uu^1)$ and $d(uu^k)$.
To this end, we need to draw the edges $uu^1,\ldots,uu^k$ and the bubble $H'=H-\{u\}$ obtained in the conclusion of Lemma~\ref{lem:bubble} statement~\ref{item:v-bubble}.
Then, the drawing of $H$ can be scaled down towards $u$ so as to avoid crossing the other predrawn parts of the graph.
We distinguish three cases:

\noindent\textit{Case 1:\/ $k=1$}.
The edge $uu^1$ is drawn as a straight-line segment in direction $d(uu^1)$, and the v-bubbles of the splitting sequence of $H'$ are drawn recursively.

\noindent\textit{Case 2:\/ $k=2$}.
The edges $uu^1$ and $uu^2$ and the root-path $P$ of $H'$ are drawn in~such a way that the following conditions are satisfied:
\begin{itemize}
\item the edge $uu^1$ leaves $u$ in direction $d(uu^1)$, bends, and continues to $u^1$ in direction $d(uu^2)$,
\item the edge $uu^2$ leaves $u$ in direction $d(uu^2)$, bends, and continues to $u^2$ in direction $d(uu^1)$,
\item each edge $xy$ of $P$ leaves $x$ in direction $-d(uu^1)$ and $y$ in direction $-d(uu^2)$,
\item for every edge $xy$ of $P$, the quadrilateral $Q(xy,d(uu^2),d(uu^1))$ is well defined.
\end{itemize}
Figure~\ref{fig:drawing} illustrates how to achieve such a drawing.
As a consequence, $d(uu^1)$ and $d(uu^2)$ can be assigned as leading directions to the e-bubbles of the splitting sequence of $H'$.
The drawing of $H$ is completed by drawing all bubbles of the splitting sequence of $H'$ recursively.

\noindent\textit{Case 3:\/ $k\geq 3$}.
Let $P$ denote the root-path of $H$ and $u^{k-1}x_1\ldots x_mu^k$ denote the part of $P$ between $u^{k-1}$ and $u^k$.
Choose a direction $d\in\{d(uu^1),\ldots,d(uu^k)\}$ so that $-d\notin\{d(uu^1),\ldots,d(uu^k)\}$.
The edges $uu^1,\ldots,uu^k$ and the root-path $P$ are drawn in such a way that the following conditions are satisfied:
\begin{itemize}
\item each edge $uu^i$ leaves $u$ in direction $d(uu^i)$, bends shortly after, and continues to $u^i$ in direction $d$,
\item each edge $xy$ of $P$ leaves $x$ in direction $-d$ if $x\in\{x_1,\ldots,x_m\}$ or $-d^-$ otherwise, and leaves $y$ in direction $-d^+$ if $y\in\{u^2,\ldots,u^{k-1},x_1,\ldots,x_m,u^k\}$ or $-d$ otherwise.
\item for every edge $xy$ of $P$, the quadrilateral $Q(xy,d_x^{xy},d_y^{xy})$ is well defined, where
\begin{equation*}
(d_x^{xy},d_y^{xy})=\begin{cases}
(d,d^-) & \text{if $x,y\in\{u^{k-1},x_1,\ldots,x_m,u^k\}$,} \\
(d^+,d) & \text{otherwise.}
\end{cases}
\end{equation*}
\end{itemize}
Again, Figure~\ref{fig:drawing} illustrates how to achieve such a drawing.
As a consequence, $d_y^{xy}$ and $d_x^{xy}$ can be assigned as leading directions to every e-bubble of the splitting sequence of $H'$, where $x$ and $y$ are the roots of the e-bubble (case distinction in the definition of $(d_x^{xy},d_y^{xy})$ is needed to ensure property~a).
The drawing of $H$ is completed by drawing all bubbles of the splitting sequence of $H'$ recursively.

To complete the proof for $s\geq 3$, pick any vertex $u$ of $G$ of degree $1$, assign an arbitrary direction to the edge at $u$, and continue the drawing as in Case~1.

\begin{figure}[t]
\begin{center}
\includegraphics[width=.9\textwidth]{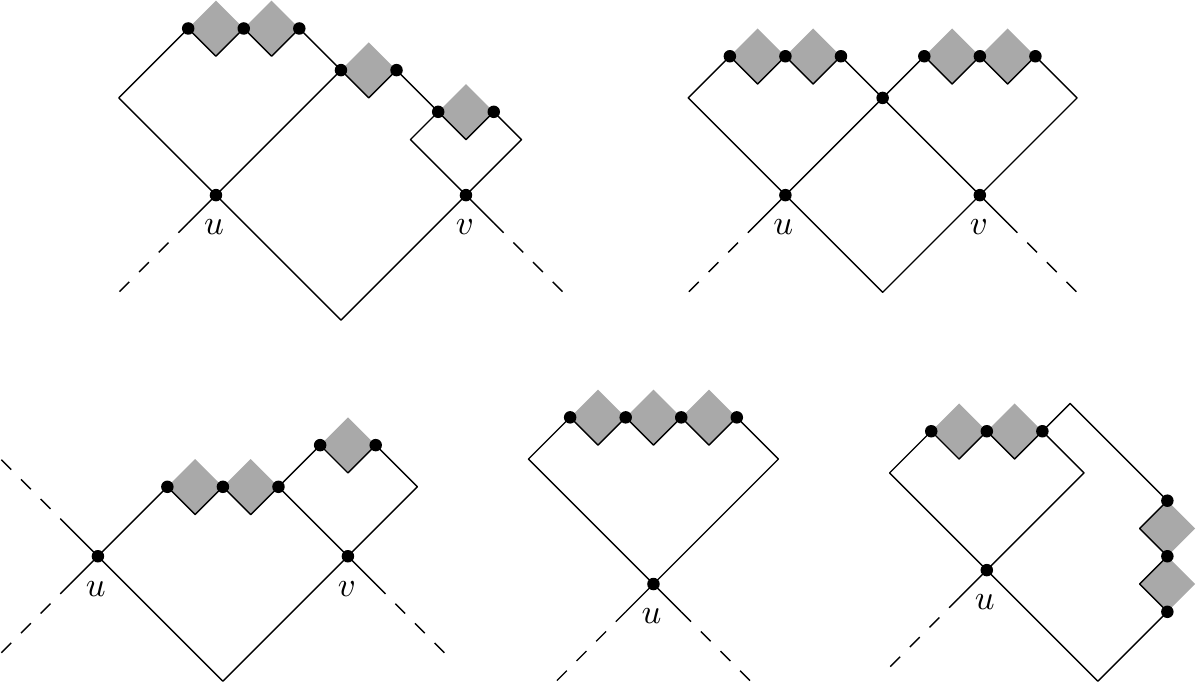}
\caption{Various ways of drawing v- and e-bubbles when $s=2$.
The target quadrilaterals for recursive e-bubbles are grayed.
The edges of the root-path of $H'$ that form trivial e-bubbles do not have target quadrilaterals.}
\label{fig:drawing2}
\end{center}
\end{figure}

The proof for $s=2$ keeps the same general recursive scheme following from Lemma~\ref{lem:bubble}.
As before, all e-bubbles are drawn in $\epsilon$-neighborhoods of their target quadrilaterals, which are always parallelograms when $s=2$.
The details of the drawing algorithm for various possible cases should be clear from Fig.~\ref{fig:drawing2}.
\end{proof}

\section{Planar Graphs and Planar Bipartite Graphs}

Using contact representations as in \cite[Theorem~2]{Kes-13}, where the upper bound of $2\Delta$ on the planar one-bend slope number is shown for planar graphs, we improve the upper bounds on this parameter for planar and bipartite planar graphs.

\begin{proposition}\label{prop:planar}
Every planar graph with maximum degree\/ $\Delta$ admits a planar one-bend drawing using at most\/ $\Delta+\lceil\frac{\Delta}{2}\rceil-1$ slopes.
\end{proposition}

\begin{proof}
Let $G$ be a graph as in the statement.
By \cite[Theorem 4.1]{Fra-94}, $G$ can be represented as a contact graph of T-shapes in the plane.
Every T-shape consists of a horizontal segment and a vertical segment touching at the upper endpoint of the vertical one.
That point, called the \emph{center} of the T-shape, splits the horizontal segment into the \emph{left segment} and the \emph{right segment} of the T-shape.
The T-shapes of the contact representation are modified as follows: for each T-shape, considered one by one in the top-down order of horizontal segments, move its vertical segment horizontally so as to make its left segment and its right segment contain at most $\lceil\frac{\Delta}{2}\rceil$ contact points with other T-shapes, and scale accordingly the two bottomless rectangular stripes going down from the left and the right segment.
This keeps the contact graph unchanged.

We construct a one-bend drawing of $G$ using a set $S_H$ of $\smash[t]{\lceil\frac{\Delta}{2}\rceil}$ almost horizontal slopes and a set $S_V$ of $\Delta-1$ almost vertical slopes.
We place each vertex $v$ at the center of the T-shape representing $v$ unless all contact points of the T-shape lie on the vertical segment.
In the latter case, we put $v$ on the vertical segment of the T-shape so that it splits the segment into two parts containing at most $\lceil\frac{\Delta}{2}\rceil$ contact points.
A vertex placed at the center of a T-shape emits at most $\smash[t]{\lceil\frac{\Delta}{2}\rceil}$ rays with slopes from $S_H$ towards the contact points on the left segment, at most $\smash[t]{\lceil\frac{\Delta}{2}\rceil}$ rays with slopes from $S_H$ towards the contact points on the right segment, and at most $\Delta-1$ rays with slopes from $S_V$ towards the contact points on the vertical segment.
A vertex placed on the vertical segment of a T-shape emits at most $\lceil\frac{\Delta}{2}\rceil$ rays with slopes from $S_V$ towards the contact points on either of the two parts of the segment.
For every edge of $G$, two appropriately chosen rays, one with slope from $S_H$ and one with slope from $S_V$, are joined near the corresponding contact point to form a representation of that edge in the claimed planar one-bend drawing of $G$, see Fig.~\ref{fig:planar}.
\end{proof}

\begin{figure}[t]
\centering
\subfigure[\label{fig:planar}]{\includegraphics{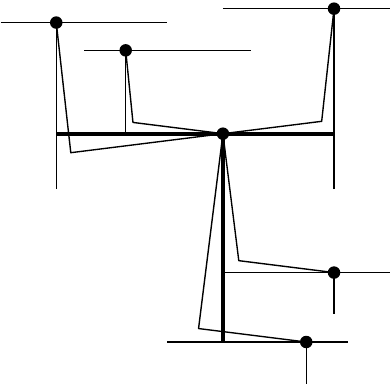}}
\hspace{1.5em}%
\subfigure[\label{fig:planarbip}]{\includegraphics{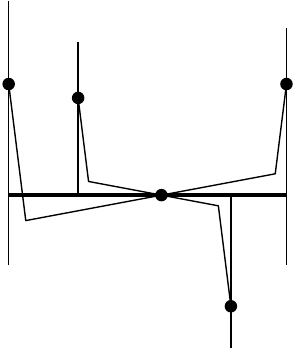}}
\hspace{1.5em}%
\subfigure[\label{fig:planarbiplb}]{\includegraphics[width=.33\textwidth]{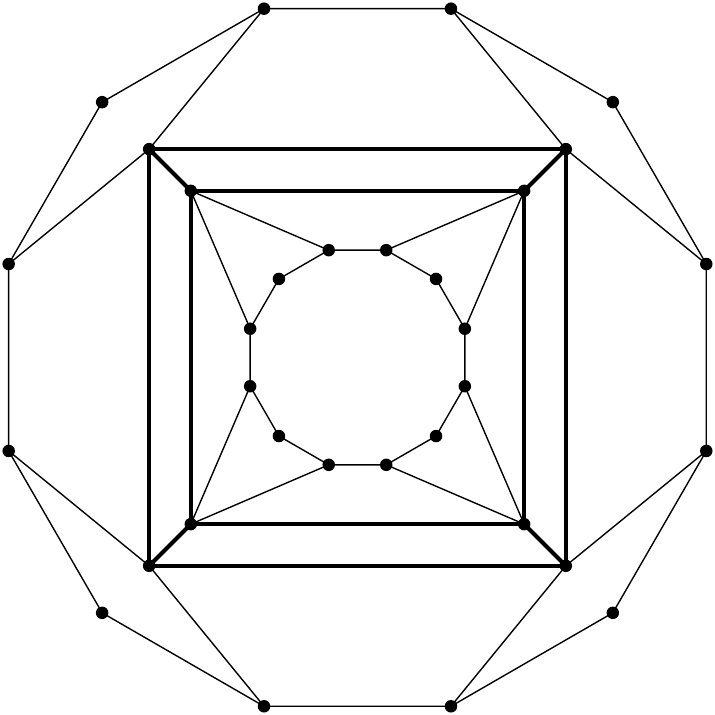}}
\caption{\subref{fig:planar} One-bend drawing of a planar graph\quad\subref{fig:planarbip} One-bend drawing of a bipartite planar graph\quad\subref{fig:planarbiplb} Graph $G_5$ constructed in the proof of Proposition~\ref{prop:lowbipplanar}}
\end{figure}

\begin{proposition}\label{prop:bipplanar}
Every bipartite planar graph with maximum degree\/ $\Delta$ admits a planar one-bend drawing using at most\/ $2\lceil\frac{\Delta}{2}\rceil$ slopes.
\end{proposition}

\begin{proof}
Let $G$ be a graph as in the statement.
By \cite[Theorem 1.5]{Fra-95}, $G$ can be represented as a contact graph of horizontal and vertical segments in the plane.
We construct a one-bend drawing of $G$ using a set $S_H$ of $\smash[t]{\lceil\frac{\Delta}{2}\rceil}$ almost horizontal slopes and a set $S_V$ of $\smash[t]{\lceil\frac{\Delta}{2}\rceil}$ almost vertical slopes.
We place every vertex $v$ of $G$ on the segment representing $v$ so that it splits the segment into two parts containing at most $\lceil\frac{\Delta}{2}\rceil$ contact points with other segments.
A vertex placed on a horizontal segment emits at most $\smash[t]{\lceil\frac{\Delta}{2}\rceil}$ rays with slopes from $S_H$ towards the contact points on either of the two parts of the segment.
Similarly, a vertex placed on a vertical segment emits at most $\smash[t]{\lceil\frac{\Delta}{2}\rceil}$ rays with slopes from $S_V$ towards the contact points on either of the two parts of the segment.
For every edge of $G$, two appropriately chosen rays, one with slope from $S_H$ and one with slope from $S_V$, are joined near the corresponding contact point to form a representation of that edge in the claimed planar one-bend drawing of $G$, see Fig.~\ref{fig:planarbip}.
\end{proof}

The following is a straightforward adaptation of \cite[Theorem~4]{Kes-13}, where planar graphs with planar one-bend slope number at least $\frac{3}{4}(\Delta-1)$ are constructed.

\begin{proposition}\label{prop:lowbipplanar}
For every\/ $\Delta\geq 3$, there is a planar bipartite graph with maximum degree\/ $\Delta$ and with planar one-bend slope number at least\/ $\frac{2}{3}(\Delta-1)$.
\end{proposition}

\begin{proof}
A graph $G_\Delta$ of maximum degree $\Delta$ is constructed starting from a plane drawing of the $3$-dimensional cube.
Two opposite faces of the cube are chosen, say, the outer and the central.
In either of them, a cycle of $8\Delta-28$ new vertices is drawn; then, each boundary vertex of the face picks a subpath of $2\Delta-7$ vertices of the cycle and connects to the $\Delta-3$ odd vertices of the subpath, see Fig.~\ref{fig:planarbiplb}.

It is well known that the measures of the interior angles of a simple $k$-gon sum up to $(k-2)\pi$.
This is a consequence of a more general observation: if $P$ is a simple $k$-gon (with angles of measure $\pi$ allowed), then every slope is covered exactly $k-2$ times by interior angles of $P$.
For the purpose of this statement, at each vertex of $P$, either of the two directions of a slope is counted separately---once if it points towards the interior of $P$ and $\frac{1}{2}$ times if it points towards the boundary of $P$.
Therefore, if $S$ is a set of slopes and $P$ is a simple $k$-gon drawn using slopes from $S$, then every slope from $S$ can be used by at most $k-2$ segments that are sides of $P$ or go from a vertex of $P$ towards the interior of $P$.

Suppose we are given a planar one-bend drawing of $G_\Delta$ using a set of slopes $S$.
The restriction of the drawing to the starting cube must have one of the two selected faces, call it $F$, as an inner face.
The face $F$ is drawn as a simple octagon (with angles of measure $\pi$ allowed), and each of the four vertices of the cube that lie on the boundary of $F$ emits $\Delta-3$ edges towards the interior of $F$.
By the observation above, every slope from $S$ can be used by at most $6$ of the $8+4(\Delta-3)$ segments that are sides of the octagon or initial parts of the edges going from the four vertices towards the interior of $F$.
We conclude that $8+4(\Delta-3)\leq 6|S|$, which yields $|S|\geq\frac{2}{3}(\Delta-1)$.
\end{proof}

\section{General Graphs}

The main contribution of this section is to show the following:

\begin{theorem}\label{thm:general}
Every graph with maximum degree\/ $\Delta$ admits a one-bend drawing using at most\/ $\lceil\frac{\Delta}{2}\rceil+1$ slopes.
Such a drawing exists with all vertices placed on a common line.
Furthermore, the set of slopes can be prescribed arbitrarily.
\end{theorem}

\begin{proof}
Let $G$ be a graph with maximum degree $\Delta$.
By Vizing's theorem~\cite{Viz-64}, $G$ has a proper edge-coloring using at most $\Delta+1$ colors, and moreover, such a coloring can be obtained in polynomial time~\cite{Mis-92}.
This yields a partition of the edge set of $G$ into $\Delta+1$ matchings $M_1,\ldots,M_{\Delta+1}$.
Let $n=|V(G)|-|M_{\Delta+1}|$, and let $f\colon V(G)\to\{1,\ldots,n\}$ be such that $f(u)=f(v)$ if and only if $uv\in M_{\Delta+1}$.

Let $S$ be a set of $k=\lceil\frac{\Delta}{2}\rceil+1$ slopes and $\ell$ be a line with slope not in $S$.
Without loss of generality, we can assume that $\ell$ is horizontal.
Order $S$ as $\{s_1,\ldots,s_k\}$ clockwise starting from the horizontal slope (that is, if $i<j$, then $s_i$ occurs before $s_j$ when rotating a line clockwise starting from the horizontal position).
Fix $n$ pairwise disjoint segments $I_1,\ldots,I_n$ in this order on $\ell$.

Each vertex $v$ of $G$ is placed on the segment $I_{f(v)}$.
Each edge $uv\in M_i$ with $1\leq i\leq k-1$ and $f(u)<f(v)$ is drawn above $\ell$ so that its slope at $v$ is $s_i$ and its slope at $u$ is $s_j$, where $j$ is the least index in $\{i+1,\ldots,k-1\}$ for which there is no edge $u'u\in M_j$ with $f(u')<f(u)$, or $j=k$ if such an index does not exist.
This way, since $M_1,\ldots,M_{k-1}$ are matchings, no two edges of $M_1,\ldots,M_{k-1}$ use the same slope at any vertex.
The edges of $M_k,\ldots,M_\Delta$ are drawn in an analogous way below $\ell$.
At least one slope above $\ell$ and at least one below $\ell$ are left free at every vertex.

Now, consider an edge $uv\in M_{\Delta+1}$.
In the drawing presented above, $u$ and $v$ have degree at most $\Delta-1$, so either of them has an additional free slope above or below $\ell$.
Therefore, either above or below $\ell$, there are two distinct slopes, one free at $u$ and the other free at $v$.
They can be used to draw the edge $uv$ if $u$ and $v$ are placed in an appropriate order within $I_{f(u)}=I_{f(v)}$.
Occurrence of bend points of some edges on other edges can be fixed by perturbing vertices slightly within their segments on $\ell$.
\end{proof}

The results of~\cite{Fel-14} directly yield a characterization of the graphs that require $\lceil\frac{\Delta}{2}\rceil+1$ slopes for a one-bend drawing when $\Delta\leq 4$.
We do not know of any graph that would require $\lceil\frac{\Delta}{2}\rceil+1$ slopes for a one-bend drawing when $\Delta\geq 5$.

\section*{Acknowledgment}

We thank Piotr Micek for fruitful discussions.

\bibliographystyle{splncs03}
\bibliography{slopebib}

\end{document}